\documentclass{article}

     \PassOptionsToPackage{numbers, compress}{natbib}

\usepackage[preprint]{neurips_2020}

\usepackage{graphicx,scalerel}
\usepackage[utf8]{inputenc} 
\usepackage[T1]{fontenc}    
\usepackage{hyperref}       
\usepackage{url}            
\usepackage{booktabs}       
\usepackage{amsfonts}       
\usepackage{nicefrac}       
\usepackage{microtype}      
\usepackage{authblk}
\usepackage{amsthm}
\usepackage{amsmath}
\usepackage{tabu}
\usepackage{lmodern}
\usepackage{amssymb}
\usepackage{subcaption}

\usepackage{scalerel}
\usepackage{graphicx}
\usepackage{epstopdf}
\usepackage{esint}
\usepackage{etoolbox}
\newcommand\wh[1]{\hstretch{2}{\hat{\hstretch{.5}{#1}}}}
\newtheorem{theorem}{Theorem}
\newtheorem{lemma}{Lemma}
\newtheorem{remark}{Remark}

\newtheorem{assumption}{Assumption}
\newtheorem{corollary}{Corollary}
\theoremstyle{definition}

\usepackage{algorithm}
\usepackage{multirow}
\usepackage{algpseudocode}
\usepackage{dsfont}

\hypersetup{
  colorlinks,
  citecolor=green,
  linkcolor=red,
  urlcolor=blue}
  
\title{Federated Learning With Quantized Global Model Updates}

\author{%
\textbf{Mohammad Mohammadi Amiri}$^1$, \textbf{Deniz G\"und\"uz}$^2$, \textbf{Sanjeev R. Kulkarni}$^1$, \textbf{H.~Vincent~Poor}$^1$\\\vspace{-.3cm}
Princeton University$^1$, Imperial College London$^2$\\
\texttt{\{mamiri, kulkarni, poor\}@princeton.edu}, \texttt{d.gunduz@imperial.ac.uk}\\
}

\begin{document}

\maketitle

\begin{abstract}
We study federated learning (FL), which enables mobile devices to utilize their local datasets to collaboratively train a global model with the help of a central server, while keeping data localized. At each iteration, the server broadcasts the current global model to the devices for local training, and aggregates the local model updates from the devices to update the global model. Previous work on the communication efficiency of FL has mainly focused on the aggregation of model updates from the devices, assuming perfect broadcasting of the global model. In this paper, we instead consider broadcasting a compressed version of the global model. 
This is to further reduce the communication cost of FL, which can be particularly limited when the global model is to be transmitted over a wireless medium. We introduce a lossy FL (LFL) algorithm, in which both the global model and the local model updates are quantized before being transmitted. 
We analyze the convergence behavior of the proposed LFL algorithm assuming the availability of accurate local model updates at the server. 
Numerical experiments show that the proposed LFL scheme, which quantizes the global model update (with respect to the global model estimate at the devices) rather than the global model itself, significantly outperforms other existing schemes studying quantization of the global model at the PS-to-device direction. 
Also, the performance loss of the proposed scheme is marginal compared to the fully lossless approach, where the PS and the devices transmit their messages entirely without any quantization.
\end{abstract}

\section{Introduction}\label{SecIntroduction}

Federated learning (FL) enables wireless devices to collaboratively train a global model by utilizing locally available data and computational capabilities under the coordination of a parameter server (PS) while the data never leaves the devices \cite{McMahanFLOnline}.

In FL with $M$ devices the goal is to minimize a loss function
$F(\boldsymbol{\theta}) = \sum\nolimits_{m=1}^{M} \frac{B_m}{B} F_m \left(\boldsymbol{\theta} \right)$
with respect to the global model $\boldsymbol{\theta} \in \mathbb{R}^d$, where $F_m \left( \boldsymbol{\theta} \right) = \frac{1}{B_m} \sum\nolimits_{\boldsymbol{u} \in \mathcal{B}_m} f \left(\boldsymbol{\theta}, \boldsymbol{u} \right)$ is the loss function at device $m$,
with $\mathcal{B}_m$ representing device $m$'s local dataset of size $B_m$, $B \triangleq \sum\nolimits_{m=1}^{M} B_m$, and $f(\cdot, \cdot)$ is an empirical loss function. 
Having access to the global model $\boldsymbol{\theta}$, device $m$ utilizes its local dataset and performs multiple iterations of stochastic gradient descent (SGD) in order to minimize the local loss function $F_m \left( \boldsymbol{\theta} \right)$.
It then sends the local model update to the server, which aggregates the local updates from all the devices to update the global model.

FL mainly targets mobile applications at the network edge, and the wireless communication links connecting these devices to the network are typically limited in bandwidth and power, and suffer from various channel impairments such as fading, shadowing, or interference; 
hence the need to develop an FL framework with limited communication requirements becomes more vital. 
While communication-efficient FL has been widely studied, prior works mainly focused on the devices-to-PS links, assuming perfect broadcasting of the global model to the devices at each iteration.
In this paper, we design an FL algorithm aiming to reduce the cost of both PS-to-device and devices-to-PS communications. 
To address the importance of quantization at the PS-to-device direction, we highlight that some devices simply may not have the sufficient bandwidth to receive the global model update when the model size is relatively large, particularly in the wireless setting, where the devices are away from the base station. 
This would result in consistent exclusion of these devices, resulting in significant performance loss.
Moreover, the impact of quantization in the device-to-PS direction is less severe due to the impact of averaging local updates at the PS.


\vspace{-0cm}
\paragraph{Related work}
There is a fast-growing body of literature on the communication efficiency of FL targeting restricted bandwidth devices. 
Several studies address this issue by considering communications with rate limitations, and propose
different compression and quantization techniques
\cite{DCKonecnyFederated,McMahan2017CommunicationEfficientLO,KonecnyRandDistMean,CryptoNetsDowlinFL,KonecnyFLBeyondData,DCLinHanDeepGradComp,COLAFLHe,MohammadDenizDSGDCS}, as well as performing local updates to reduce the frequency of communications from the devices to the PS \cite{UseLocalSGDLin,LocalSGDConvFastStich}. 
Statistical challenges arise in FL since the data samples may not be independent and identically distributed (iid) across devices. The common sources of the dependence or bias in data distribution are the participating devices being located in a particular geographic region, and/or at a particular time window \cite{AdvancesOpenProbsFL}. Different approaches have been studied to mitigate the effect of non-iid data in FL \cite{McMahan2017CommunicationEfficientLO,HsiehNonIIDQuagmire,FLHeterogeneousNetsLi,DatasetDistillationWang,Semi_CyclicSGD,FLWithnonIIDZhao}. 
Also, FL suffers from a significant variability in the system, which is mainly due to the hardware, network connectivity, and available power associated with different devices \cite{FLChallMethFutDir}. 
Active device selection schemes have been introduced to alleviate significant variability in FL systems, where a subset of devices share the resources and participate at each iteration of training \cite{IncentiveFLKang,ClientSelectionFLNishio,MohammadDenizSanjVinceISIT20,HowardVinceSchedulingsFL,YangArafaVinceAgeBasedFL}.
There have also been efforts in developing convergence guarantees for FL under various scenarios, considering iid data across the devices \cite{LocalSGDConvFastStich,CooperativeSGDFLJoshi,GraphOracleModelsWoodworth,ConverKStepSGDZhou,UnifiedTheoryDecDGDTopologyKoloskova}, non-iid data \cite{UnifiedTheoryDecDGDTopologyKoloskova,FLHeterogeneousNetsLi,ConvLocalDescentFLHaddadpour,XLiFedAveFLnonIID}, participation of all the devices \cite{FirstConvAnalysisGDKhaled,AdaptiveFLEdgeComputingWang,ParallelRestartedSGDYu,FasterOnDeviceFLMomentum}, or only a subset of devices at each iteration \cite{FedDANENewtonTypeFL,SCAFFOLDControlledFL,DynamicFederatedLearningSayed,XLiFedAveFLnonIID,FLTWCConvergenceMohammadDenizSanjVince}, and FL under limited communication constraints \cite{FLTWCConvergenceMohammadDenizSanjVince,DCRechtHogwildSGD,SparseConvergenceAlistarh,FedPAQFLPeriodicQuantizationReisizadeh}.
Furthermore, FL with compressed global model transmission has been studied recently in \cite{ExpandingRedClientResFL,DoubleSqueezeTangConf} aiming to alleviate the communication footprint from the PS to the devices.
Since the global model parameters are relatively skewed/diverse, with the scheme in \cite{ExpandingRedClientResFL} at each iteration the PS employs a linear transform before quantization, and the devices apply the inverse linear transform to estimate the global model. 
On the other hand, error compensation at the PS is employed in \cite{DoubleSqueezeTangConf} to accumulate the error of quantizing the global model.

\vspace{-0cm}
\paragraph{Our contributions} 
With the exception of \cite{ExpandingRedClientResFL,DoubleSqueezeTangConf}, the literature on FL considers perfect broadcasting of the global model from the PS to the devices. 
With this assumption, no matter what type of local update or device-to-PS communication strategy is used, all the devices are synchronized with the same global model at each iteration.
In this paper, we instead consider broadcasting a quantized version of the global model update by the PS, which provides the devices with a lossy estimate of the global model (rather than its accurate estimate) with which to perform local training. 
This further reduces the communication cost of FL, which can be particularly limited for transmission over a wireless medium while serving a massive number of devices. 
Also, it is interesting to investigate the impact of various hyperparameters on the performance of FL with lossy broadcasting of the global model since FL involves transmission over wireless networks with limited bandwidth. 
We introduce a lossy FL (LFL) algorithm, where at each iteration the PS broadcasts a compressed version of the global model update to all the devices through quantization.
To be precise, the PS exploits the knowledge of the last global model estimate available at the devices as side information to quantize the global model update.
The devices recover an estimate of the current global model by combining the received quantized global model update with their previous estimate, and perform local training using their estimate, and return the local model updates, again employing quantization.
The PS updates the global model after receiving the quantized local model updates from the devices. 
We provide convergence analysis of the LFL algorithm investigating the impact of lossy broadcasting on the performance of FL, where for ease of analysis we assume the availability of accurate local model updates from the devices at the PS. 
Numerical experiments on the MNIST and CIFAR-10 datasets illustrate the efficiency of the proposed LFL algorithm. 
We observe that the proposed LFL scheme, which leads to a significant communication cost saving, provides a promising performance with no visible gap to the performance of the fully lossless scenario where the communication from both PS-to-device and device-to-PS directions is assumed to be perfect. 
Also, it is illustrated that the proposed LFL scheme significantly outperforms the schemes introduced in \cite{ExpandingRedClientResFL} and \cite{DoubleSqueezeTangConf} considering compression from the PS to devices. 

We highlight that the proposed LFL algorithm differs from the approaches introduced in \cite{ExpandingRedClientResFL,DoubleSqueezeTangConf}, where the PS sends a quantized version of the current global model to a subset of devices that will participate in the learning process at that iteration. 
The efficiency of quantization diminishes significantly when the peak-to-average ratio of the parameters is large. 
To overcome this, in \cite{ExpandingRedClientResFL} the PS first employs a linear transform in order to spread the information of the global model vector more evenly among its dimensions, and broadcasts a quantized version of the resultant vector.
Furthermore, in \cite{DoubleSqueezeTangConf} the PS broadcasts quantized global model with error accumulation to compensate the quantization error. 
Instead, we propose broadcasting the global model update, with respect to the previous estimate at the devices, rather than the global model itself.
We remark that the global model update
has less variability/variance and peak-to-average ratio than the global model, and hence, for the same communication load, the devices can have a more accurate estimate of the global model.
However, this would require all the devices to track the global model at each iteration, even if they do not participate in the learning process by sending their local update.
We argue that broadcasting the global model update to the whole set of devices, rather than a randomly chosen subset, would introduce limited additional communication cost as broadcasting is typically more efficient than sending independent information to devices. 
Moreover, in practice, the subset of participating devices remain the same for a number of iterations, until a device leaves or joins.
Our algorithm can easily be adopted to such scenarios by sending the global model, rather than the model update, every time the subset of devices changes.
Note also that, compared to the LFL algorithm, the approach introduced in \cite{ExpandingRedClientResFL} requires a significantly higher computational overhead due to employing the linear transform at the PS and its inverse at the devices, where this overhead grows with the size of the model parameters. 
Furthermore, the performance evaluation in \cite{ExpandingRedClientResFL} is limited to the experimental results, while in this paper we provide an in-depth convergence analysis of the proposed LFL algorithm. 
The advantage of the proposed LFL algorithm over the approaches introduced in \cite{ExpandingRedClientResFL,DoubleSqueezeTangConf} is shown numerically, where, despite its significantly smaller communication load, it provides considerably higher accuracy. This illustrates that quantizing the global model update provides a more accurate global model estimate at the devices than quantizing the global model itself.



\vspace{-0cm}
\paragraph{Notation} The set of real numbers is denoted by $\mathbb{R}$. For $x \in \mathbb{R}$, $\left| x \right|$ returns the absolute value of $x$. 
For a vector of real numbers $\boldsymbol{x}$, the largest and the smallest absolute values among all the entries of $\boldsymbol{x}$ are represented by $\max \left\{ \left| \boldsymbol{x} \right| \right\}$ and $\min \left\{ \left| \boldsymbol{x} \right| \right\}$, respectively. 
For an integer $i$, we let $[i] \triangleq \{1, 2, \dots, i\}$. The $l_2$-norm of vector $\boldsymbol{x}$ is denoted by $\left\| \boldsymbol{x} \right\|_2$.

\vspace{-0.2cm}
\section{Lossy Federated Learning (LFL) Algorithm}
\label{SecLFL}
We consider a lossy PS-to-device transmission, in which the PS sends a compressed version of the global model to the devices.
This reduces the communication cost, and can be particularly beneficial when the PS resources are limited, and/or communication takes place over a constrained bandwidth medium.
We denote the estimate of the global model ${\boldsymbol{\theta}} (t)$ at the devices by $\wh{\boldsymbol{\theta}} (t)$, where $t$ represents the global iteration count. 
Having recovered $\wh{\boldsymbol{\theta}} (t)$, the devices perform a $\tau$-step SGD with respect to their local datasets, and transmit their local model updates to the PS using quantization while accumulating the quantization error.

\vspace{-0cm}
\subsection{Global Model Broadcasting}\label{SubSecBroadcast}
In the proposed LFL algorithm, the PS performs stochastic quantization similarly to the QSGD algorithm introduced in \cite{DCAlistarhQSGD} with a slight modification to broadcast the information about the global model to the devices. 
In particular, at global iteration $t$, the PS aims to broadcast the global model update $\boldsymbol{\theta} (t) - \wh{\boldsymbol{\theta}} (t-1)$ to the devices.
We present the stochastic quantization technique we use, denoted by $Q(\cdot, \cdot)$, in Appendix \ref{AppStochQuan}.

\begin{lemma}\label{LemDigQuanVar}
For the quantization function $\varphi \left( x, q \right)$ and vector $\boldsymbol{Q}(\boldsymbol{x}, q)$ given in (\ref{DefQ}) and (\ref{QVecDefinition}), respectively, we have
\begin{subequations}
\begin{align}\label{Expectation_Lemma_phi}
& \mathbb{E}_{\varphi} \left[  \varphi(x, q) \right] = x, \quad
\mathbb{E}_{\varphi} \big[ \varphi^2(x, q) \big] \le x^2 + 1/(4q^2), \\ 
& \mathbb{E}_{\varphi} \left[  \boldsymbol{Q}(\boldsymbol{x}, q) \right] = \boldsymbol{x}, \quad \mathbb{E}_{\varphi} \big[  \left\| \boldsymbol{Q}(\boldsymbol{x}, q) \right\|^2_2 \big] \le \left\| \boldsymbol{x} \right\|_2^2 + \varepsilon d \left\| \boldsymbol{x} \right\|_2^2/(4 q^2),
\end{align}
\end{subequations}
where $\mathbb{E}_{\varphi}$ represents expectation with respect to the quatization function $\varphi \left( \cdot, \cdot \right)$, and $0 \le \varepsilon \le 1$ is defined as $\varepsilon \triangleq \left( \max \left\{ \left| \boldsymbol{x} \right| \right\} - \min \left\{ \left| \boldsymbol{x} \right| \right\} \right)^2 /  \left\| \boldsymbol{x} \right\|_2^2$. 
\end{lemma}

The proof of Lemma \ref{LemDigQuanVar} is provided in Appendix \ref{AppProofLemma}.
We highlight that the value of $\varepsilon$ depends on the skewness of the magnitudes of the entries of $\boldsymbol{x}$, where it increases for a more skewed entries with a higher variance.  
We have $\varepsilon = 0$, if and only if all the entries of $\boldsymbol{x}$ have the same magnitude, and $\varepsilon = 1$, if and only if $\boldsymbol{x}$ has only one non-zero entry.

\begin{algorithm}[t]
\caption{LFL}
\label{ModelUpdateAlg}
\begin{algorithmic}[1]
\Statex
\For {$t = 0, \ldots, T-1$}
\Statex
\begin{itemize}
\item \textbf{Global model broadcasting}
\end{itemize}
\State{PS broadcasts $\boldsymbol{Q}\big(\boldsymbol{\theta} (t) - \wh{\boldsymbol{\theta}} (t-1), q_1\big)$}
\State{$\wh{\boldsymbol{\theta}} (t) = \wh{\boldsymbol{\theta}} (t-1) + \boldsymbol{Q}\big(\boldsymbol{\theta} (t) - \wh{\boldsymbol{\theta}} (t-1), q_1\big)$}
\Statex
\begin{itemize}
\item \textbf{Local update aggregation}
\end{itemize}
\For {$m = 1, \dots, M$ in parallel}
\State{Device $m$ transmits $\boldsymbol{Q}\big( \Delta \boldsymbol{\theta}_m (t) + \boldsymbol{\delta}_m (t), q_2 \big) = \boldsymbol{Q}\big( \boldsymbol{\theta}_m^{\tau+1} (t) - \wh{\boldsymbol{\theta}} (t) + \boldsymbol{\delta}_m (t), q_2 \big)$}
\EndFor
\State{$\boldsymbol{\theta} (t+1) = \wh{\boldsymbol{\theta}} (t) + \sum\nolimits_{m = 1}^M \frac{B_m}{B} \boldsymbol{Q}\big( \Delta \boldsymbol{\theta}_m (t) + \boldsymbol{\delta}_m (t), q_2 \big)$}
\EndFor
\end{algorithmic}\label{Dig_Downlink}
\end{algorithm}

Given a quantization level $q_1$, the PS broadcasts $\boldsymbol{Q}\big(\boldsymbol{\theta} (t) - \wh{\boldsymbol{\theta}} (t-1), q_1\big)$ to the devices at global iteration $t$. 
Then the devices obtain the following estimate of $\boldsymbol{\theta} (t)$:
\begin{align}\label{ThetaHat_t}
\wh{\boldsymbol{\theta}} (t) = \wh{\boldsymbol{\theta}} (t-1) + \boldsymbol{Q}\big(\boldsymbol{\theta} (t) - \wh{\boldsymbol{\theta}} (t-1), q_1\big),   
\end{align}
which is equivalent to $\wh{\boldsymbol{\theta}} (t) = {\boldsymbol{\theta}} (0) + \sum\nolimits_{i=1}^{t} \boldsymbol{Q}\big(\boldsymbol{\theta} (i) - \wh{\boldsymbol{\theta}} (i-1), q_1\big)$,
where we assumed that $\wh{\boldsymbol{\theta}} (0) = {\boldsymbol{\theta}} (0)$.
We note that, having the knowledge of the compressed vector $\boldsymbol{Q}\big(\boldsymbol{\theta} (i) - \wh{\boldsymbol{\theta}} (i-1), q_1\big)$, $\forall i \in [t]$, the PS can also track $\wh{\boldsymbol{\theta}} (t)$ at each iteration.

\subsection{Local Update Aggregation}\label{SubSecAggregation} 

After recovering $\wh{\boldsymbol{\theta}} (t)$, device $m$ performs a $\tau$-step local SGD, where the $i$-th step corresponds to 
$\boldsymbol{\theta}_m^{i+1} (t) = \boldsymbol{\theta}_m^i (t) - \eta^i_m (t) \nabla F_m \left( \boldsymbol{\theta}_m^i (t), \xi_m^i (t) \right)$, $i \in [\tau]$,
where $\boldsymbol{\theta}_m^1 (t) = \wh{\boldsymbol{\theta}} (t)$, and $\xi_m^i (t)$ denotes the local mini-batch chosen uniformly at random from the local dataset $\mathcal{B}_m$.
It then aims to transmit local model update $\Delta \boldsymbol{\theta}_m (t) = \boldsymbol{\theta}_m^{\tau+1} (t) - \wh{\boldsymbol{\theta}} (t)$ through quantization with error compensation.
It performs quantization after error compensation accumulating the quantization error, and transmits $\boldsymbol{Q}\big( \Delta \boldsymbol{\theta}_m (t) + \boldsymbol{\delta}_m (t), q_2 \big)$ using a quantization level $q_2$, where $\boldsymbol{\delta}_m (t)$ retains the quantization error, and is updated as 
\begin{align}
\boldsymbol{\delta}_m (t+1) = \Delta \boldsymbol{\theta}_m (t) + \boldsymbol{\delta}_m (t) - \boldsymbol{Q}\big( \Delta \boldsymbol{\theta}_m (t) + \boldsymbol{\delta}_m (t), q_2 \big),   
\end{align}
where we set $\boldsymbol{\delta}_m (0) = \boldsymbol{0}$.
Having received $\boldsymbol{Q}\big( \Delta \boldsymbol{\theta}_m (t) + \boldsymbol{\delta}_m (t), q_2 \big)$ from device $m$, $\forall m \in [M]$, the PS updates the global model as 
\begin{align}
\boldsymbol{\theta} (t+1) = \wh{\boldsymbol{\theta}} (t) + \sum\nolimits_{m = 1}^M \frac{B_m}{B} \boldsymbol{Q}\big( \Delta \boldsymbol{\theta}_m (t) + \boldsymbol{\delta}_m (t), q_2 \big).
\end{align}

Algorithm \ref{ModelUpdateAlg} summarizes the proposed LFL algorithm.

\begin{remark}
We do not consider error compensation at the PS with LFL since we have observed performance degradation numerically when compensating the quantization error at the PS.
We argue that LFL naturally accumulates the quantization error at the PS since it sends the quatized global model update with respect to the last global model estimate at the devices.
We further highlight that the proposed approach is not limited to any specific quantization technique, and any compression technique can be used within the proposed framework. 
\end{remark}

\section{Convergence Analysis of LFL Algorithm}\label{SecConvergence}

Here we analyze the convergence behaviour of LFL, where for simplicity of the analysis, we assume that the devices can transmit their local updates, $\Delta \boldsymbol{\theta}_m (t)$, $\forall m$, accurately/in a lossless fashion to the PS, and focus on the impact of lossy broadcasting on the convergence.

\subsection{Preliminaries}\label{SebSecPreliminary}

We denote the optimal solution minimizing loss function $F(\boldsymbol{\theta})$ by $\boldsymbol{\theta}^*$, and the minimum loss as $F^*$, i.e., $\boldsymbol{\theta}^* \triangleq \arg \mathop {\min }\nolimits_{\boldsymbol{\theta}} F(\boldsymbol{\theta})$, and $F^* \triangleq F(\boldsymbol{\theta}^*)$. We also denote the minimum value of the local loss function at device $m$ by $F_m^*$. We further define $\Gamma \triangleq F^* -  \sum\nolimits_{m=1}^{M} \frac{B_m}{B} F^*_m$, where $\Gamma \ge 0$, and its magnitude indicates the bias in the data distribution across devices. 

For ease of analysis, we set $\eta_m^i (t) = \eta(t)$. Thus, the $i$-th step SGD at device $m$ is given by
\begin{align}\label{ConvSGDDevicem}
\boldsymbol{\theta}_m^{i+1} (t) = \boldsymbol{\theta}_m^i (t) - \eta (t) \nabla F_m \left( \boldsymbol{\theta}_m^i (t), \xi_m^i (t) \right),  \quad \mbox{$i \in [\tau]$}, \mbox{$m \in [M]$},   
\end{align}
where $\boldsymbol{\theta}_m^1 (t) = \wh{\boldsymbol{\theta}} (t)$, given in (\ref{ThetaHat_t}). Device $m$ transmits the local model update
\begin{align}
\Delta \boldsymbol{\theta}_m (t) = \boldsymbol{\theta}_m^{\tau+1} (t) - \wh{\boldsymbol{\theta}} (t) = - \eta (t) \sum\nolimits_{i=1}^{\tau} \nabla F_m \left( \boldsymbol{\theta}_m^i (t), \xi_m^i (t) \right), \quad  m \in [M],   
\end{align}
and the PS updates the global model as
\begin{align}\label{GlobalModelUpdateConv}
\boldsymbol{\theta} (t+1) =  \wh{\boldsymbol{\theta}} (t) - \eta (t) \sum\nolimits_{m = 1}^M \sum\nolimits_{i=1}^{\tau} \frac{B_m}{B} \nabla F_m \left( \boldsymbol{\theta}_m^i (t), \xi_m^i (t) \right).
\end{align}

\begin{assumption}\label{AssumpSmoothLoss}
The loss functions $F_1, \dots, F_M$ are $L$-smooth; that is, $\forall \boldsymbol{v}, \boldsymbol{w} \in \mathbb{R}^d$, 
\begin{align}\label{ConvLSmoothCondit}
2 \big(F_m(\boldsymbol{v}) - F_m(\boldsymbol{w}) \big) \le 2 \langle \boldsymbol{v} - \boldsymbol{w} , \nabla F_m (\boldsymbol{w}) \rangle + L \left\| \boldsymbol{v} - \boldsymbol{w} \right\|^2_2, \quad \forall m \in [M].
\end{align}
\end{assumption}

\begin{assumption}\label{AssumpStrongConvexLoss}
The loss functions $F_1, \dots, F_M$ are $\mu$-strongly convex; that is, $\forall \boldsymbol{v}, \boldsymbol{w} \in \mathbb{R}^d$, 
\begin{align}\label{ConvMuStConvexCondit}
2 \big(F_m(\boldsymbol{v}) - F_m(\boldsymbol{w}) \big) \ge 2 \langle \boldsymbol{v} - \boldsymbol{w} , \nabla F_m (\boldsymbol{w}) \rangle + \mu \left\| \boldsymbol{v} - \boldsymbol{w} \right\|^2_2, \quad \forall m \in [M].      
\end{align}
\end{assumption}

\begin{assumption}\label{AssumpBoundedVarGradient}
The expected squared $l_2$-norm of the stochastic gradients are bounded, i.e.,
\begin{align}\label{ConvNorm2Bound}
\mathbb{E}_{\xi} \left [ \left\| \nabla F_m \left( \boldsymbol{\theta}_m^i (t), \xi_m^i (t) \right) \right\|^2_2 \right] \le G^2, \quad \forall i \in [\tau], \forall m \in [M], \; \forall t.      
\end{align}
\end{assumption}

\subsection{Convergence Rate}\label{SebSecConvRate}

In the following theorem, whose proof is provided in Appendix \ref{AppProofTheorem}, we present the convergence rate of the LFL algorithm assuming that the devices can send their local updates accurately.

\begin{theorem}\label{Theoremtheta_thetastarKequalM}
Let $0 < \eta(t) \le \min \big\{ 1, \frac{1}{\mu \tau} \big\}$, $\forall t$. We have
\begin{subequations}\label{ConvTheoremtheta_thetastarKequalM}
\begin{align}\label{ConvTheoremtheta_thetastarKequalM_main}
\mathbb{E} \big[ \left\| \boldsymbol{\theta} (t) - {\boldsymbol{\theta}}^* \right\|_2^2 \big] \le  \Big( \prod\nolimits_{i=0}^{t-1} A(i) \Big) \left\| {\boldsymbol{\theta}} (0) - {\boldsymbol{\theta}}^* \right\|_2^2 + \sum\nolimits_{j=0}^{t-1} B(j) \prod\nolimits_{i=j+1}^{t-1} A(i),  
\end{align}
where 
\begin{align}\label{ConvTheoremtheta_thetastarKequalM_AB}
A(i) \triangleq & 1 - \mu \eta (i) \left( \tau - \eta(i) (\tau - 1) \right),\\ 
B(i) \triangleq & \left( 1 - \mu \eta (i) \left( \tau - \eta(i) (\tau - 1) \right) \right) \Big(\frac{\eta(i-1) \tau G}{2q_1}\Big)^2 \varepsilon d + \eta^2(i) (\tau^2 + \tau-1) G^2  \nonumber\\
&+ \left( 1+ \mu (1- \eta(i)) \right) \eta^2(i) G^2 \frac{\tau (\tau-1)(2\tau-1)}{6}   + 2 \eta(i) (\tau - 1) \Gamma,
\end{align}
\end{subequations}
for some $0 \le \varepsilon \le 1$, and the expectation is with respect to the stochastic gradient function and stochastic quantization.
\end{theorem}

\begin{corollary}
From the $L$-smoothness of the loss function, for $0 < \eta(t) \le \min \big\{ 1, \frac{1}{\mu \tau} \big\}$, $\forall t$, and a total of $T$ global iterations, it follows that \vspace{-.05cm}
\begin{align}\label{A_ConvF_FstarKequalM}
\mathbb{E} \left[ F( \boldsymbol{\theta} (T)) \right] - F^* \le & \frac{L}{2} \mathbb{E} \big[ \left\| \boldsymbol{\theta} (T) - {\boldsymbol{\theta}}^* \right\|_2^2 \big] \nonumber \\
\le & \frac{L}{2} \Big( \prod\nolimits_{i=0}^{T-1} A(i) \Big) \left\| {\boldsymbol{\theta}} (0) - {\boldsymbol{\theta}}^* \right\|_2^2 + \frac{L}{2} \sum\nolimits_{j=0}^{T-1} B(j) \prod\nolimits_{i=j+1}^{T-1} A (i), 
\end{align}\vspace{-.05cm}
where the last inequality follows from (\ref{ConvTheoremtheta_thetastarKequalM_main}).
Considering $\eta(t) = \eta$ and $\tau =1$, we have
\begin{align}
\mathbb{E} \left[ F( \boldsymbol{\theta} (T)) \right] - F^* \le & \frac{L}{2} (1-\mu \eta)^T \left\| {\boldsymbol{\theta}} (0) - {\boldsymbol{\theta}}^* \right\|_2^2 \nonumber\\
& + \frac{L}{2} \Big((1-\mu \eta) \Big( \frac{\varepsilon d}{4 q_1^2} \Big) + 1 \Big)  \big(1 - (1-\mu \eta)^T \big) \Big( \frac{\eta G^2}{\mu} \Big).    
\end{align}  
\end{corollary}

\paragraph{Asymptotic convergence analysis}
Here we show that, for a decreasing learning rate over time, such that $\mathop {\lim }\nolimits_{t \to \infty } \eta(t) = 0$, and given small enough $\varepsilon$, $\mathop {\lim }\nolimits_{T \to \infty } \mathbb{E} \left[ F( \boldsymbol{\theta} (T)) \right] - F^* = 0$.
For $0 < \eta(t) \le \min \{ 1, \frac{1}{\mu \tau} \}$, we have $0 \le A(t) < 1$, and $\mathop {\lim }\nolimits_{T \to \infty} \prod\nolimits_{i=0}^{T-1} A(i) = 0$.
For simplicity, assume $\eta(t) = \frac{\alpha}{t + \beta}$, for constant values $\alpha$ and $\beta$.
For $j \gg 0$, $B(j) \to 0$, and for limited $j$ values, $\prod\nolimits_{i=j+1}^{T-1} A(i) \to 0$, and so, according to (\ref{A_ConvF_FstarKequalM}), $\mathop {\lim }\nolimits_{T \to \infty } \mathbb{E} \left[ F( \boldsymbol{\theta} (T)) \right] - F^* = 0$.


\paragraph{Choice of $\varepsilon$}
We highlight that $\varepsilon$ appears in the convergence analysis of the LFL algorithm in inequality (\ref{App_ThetaHatWithA}), in which we have
\begin{align}\label{Ineq_For_eps}
\mathbb{E} & \left[ \Big( \max \Big\{ \Big| \sum\nolimits_{m = 1}^M \sum\nolimits_{i=1}^{\tau} \frac{B_m}{B} \nabla F_m \left( \boldsymbol{\theta}_m^i (t-1), \xi_m^i (t-1) \right) \Big| \Big\}  \Big. \right.  \nonumber\\
& \quad \left. \Big. - \min \Big\{ \Big| \sum\nolimits_{m = 1}^M \sum\nolimits_{i=1}^{\tau} \frac{B_m}{B} \nabla F_m \left( \boldsymbol{\theta}_m^i (t-1), \xi_m^i (t-1) \right) \Big| \Big\}  \Big)^2 \right] \nonumber\\
& \le \varepsilon \mathbb{E} \left[ \Big\| \sum\nolimits_{m = 1}^M \sum\nolimits_{i=1}^{\tau} \frac{B_m}{B} \nabla F_m \left( \boldsymbol{\theta}_m^i (t-1), \xi_m^i (t-1) \right) \Big\|_2^2 \right],
\end{align}
which follows from (\ref{Quantization_SquareExpectation}), where we note that \vspace{-.05cm}
\begin{align}\label{Theta_ThetaHat}
\boldsymbol{\theta} (t) - \wh{\boldsymbol{\theta}} (t-1) = - \eta (t-1) \sum\nolimits_{m = 1}^M \sum\nolimits_{i=1}^{\tau} \frac{B_m}{B} \nabla F_m \left( \boldsymbol{\theta}_m^i (t-1), \xi_m^i (t-1) \right).    
\end{align} \vspace{-.05cm}
On average the entries of $\boldsymbol{\theta} (t) - \wh{\boldsymbol{\theta}} (t-1)$, given in (\ref{Theta_ThetaHat}), are not expected to have very diverse magnitudes. Thus, the inequality in (\ref{Ineq_For_eps}) should hold for a relatively small value of $\varepsilon$. We have observed numerically that $\varepsilon \approx 10^{-3}$ satisfies inequality (\ref{Ineq_For_eps}) for the LFL algorithm.      

\vspace{-.25cm}
\paragraph{Impact of lossy broadcasting} The first term in $B(i)$ is due to the imperfect broadcasting of the global model update at the PS, which decreases with $q_1$ and increases linearly with $\varepsilon$.
This term is a complicated function of the number of $\tau$ depending on other setting variables. 




\section{Numerical Experiments}\label{SecExperiments}
Here we investigate the performance of the proposed LFL algorithm for image classification on both  MNIST \cite{LeCunMNIST} and CIFAR-10 \cite{Cifar10DatasetKrizhevsky} datasets
utilizing ADAM optimizer \cite{ADAMDC}.
We consider $M=40$ devices, and we measure the performance as the accuracy with respect to the test samples, called \textit{test accuracy}.

\paragraph{Network architecture} We train different convolutional neural networks (CNNs) with MNIST and CIFAR-10 datasets. The architectures of these CNNs are described in Table \ref{TableCNNArchit}.

\begin{table}[t!]
\caption{CNN architecture for image classification on MNIST and CIFAR-10.}
\centering
\begin{tabular}{|c||c|}
\hline
MNIST & CIFAR-10 \\ \specialrule{.2em}{.01em}{.01em}
\multirow{3}{*}{\begin{tabular}[c]{@{}c@{}}3 $\times$ 3 convolutional layer, 32 channels, \\ ReLU activation, same padding\end{tabular}} & \begin{tabular}[c]{@{}c@{}}3 $\times$ 3 convolutional layer, 32 channels, \\ ReLU activation, same padding\end{tabular}  \\ \cline{2-2} & \begin{tabular}[c]{@{}c@{}}3 $\times$ 3 convolutional layer, 32 channels, \\ ReLU activation, same padding\end{tabular} \\ \hline
\multicolumn{2}{|c|}{2 $\times$ 2 max pooling} \\ \hline
\multirow{5}{*}{\begin{tabular}[c]{@{}c@{}}3 $\times$ 3 convolutional layer, 64 channels, \\ ReLU activation, same padding\end{tabular}} & dropout with probability 0.2 \\ \cline{2-2} & \begin{tabular}[c]{@{}c@{}}3 $\times$ 3 convolutional layer, 64 channels, \\ ReLU activation, same padding\end{tabular} \\ \cline{2-2} & \begin{tabular}[c]{@{}c@{}}3 $\times$ 3 convolutional layer, 64 channels, \\ ReLU activation, same padding\end{tabular} \\ \hline
\multicolumn{2}{|c|}{2 $\times$ 2 max pooling} \\ \hline
\multirow{5}{*}{\begin{tabular}[c]{@{}c@{}}3 $\times$ 3 convolutional layer, 64 channels, \\ ReLU activation, same padding\end{tabular}} & dropout with probability 0.3 \\ \cline{2-2} & \begin{tabular}[c]{@{}c@{}}3 $\times$ 3 convolutional layer, 128 channels, \\ ReLU activation, same padding\end{tabular} \\ \cline{2-2} & \begin{tabular}[c]{@{}c@{}}3 $\times$ 3 convolutional layer, 128 channels, \\ ReLU activation, same padding\end{tabular} \\ \hline
\multicolumn{2}{|c|}{2 $\times$ 2 max pooling} \\ \hline
\begin{tabular}[c]{@{}c@{}}fully connected layer with 128 units, \\ ReLU activation\end{tabular} & dropout with probability 0.4 \\ \hline
\multicolumn{2}{|c|}{softmax output layer with 10 units} \\ \hline
\end{tabular}
\label{TableCNNArchit}
\end{table}

\paragraph{Data distribution} We consider two data distribution scenarios. 
In the non-iid scenario, we split the training data samples with the same label (from the same class) to $M/10$ disjoint subsets (assume that $M$ is divisible by 10). 
We then assign each subset of data samples, selected at random, to a different device.
In the iid scenario, we randomly split the training data samples to $M$ disjoint subsets, and assign each subset to a distinct device. 
We consider non-iid and iid data distributions while training using MNIST and CIFAR-10, respectively.


\paragraph{State-of-the-art approaches} We consider two approaches with lossy broadcasting introduced in \cite{ExpandingRedClientResFL} and \cite{DoubleSqueezeTangConf} as the state-of-the-art approaches.
With the scheme in \cite{ExpandingRedClientResFL}, referred to as lossy transformed global model (LTGM), the PS first employs a linear transform to project the global model. 
It then quantizes the resultant vector after the linear transform, and sends the quantized vector to the devices. 
The devices employ the inverse of the linear transform and use the recovered vector for local training. 
As suggested in \cite{ExpandingRedClientResFL}, we consider Walsh-Hadamrd transform and employ the stochastic quantization scheme presented in Appendix \ref{AppStochQuan} at the PS.
On the other hand, with the approach studied in \cite{DoubleSqueezeTangConf}, referred to as lossy global model (LGM), the PS directly quantizes the global model plus the quantization error accumulated from the previous iterations and shares the quantized global model with the devices, while updating the qunatization error. 
For fairness, we consider the quantization scheme presented in Appendix \ref{AppStochQuan} with the LGM scheme, and assume the same technique for transmission in the device-to-PS direction introduced in Section \ref{SubSecAggregation}.

\paragraph{Benchmark approaches} We consider the performance of the lossless broadcasting (LB) scenario, where the devices receive the current global model accurately, and perform the quantization with error compensation approach as described in Section \ref{SubSecAggregation}. 
We highlight that this approach requires transmission of $R_{\rm{LB}} = 33d$ bits from the PS, where we assume that each entry of the global model is represented by $33$ bits. 
Thus, the saving ratio in the communication bits of broadcasting from the PS using LFL versus LB is
\newcommand\approximate{\mathrel{\overset{\makebox[0pt]{\mbox{\normalfont\tiny\sffamily (a)}}}{\approx}}}
\begin{align}
\frac{R_{\rm{LB}}}{R_{\rm{Q}}} = \frac{33d}{64 + d \left( 1 + \log_2(q_1+1) \right)} \approximate  \frac{33}{1 + \log_2(q_1+1)},  
\end{align}
where (a) follows assuming that $d \gg 1$.
We further consider the performance of the fully lossless approach, where in addition to having the accurate global model at the devices, we assume that the PS receives the local model updates from the devices accurately.

\begin{figure}[t!]
\centering
\begin{subfigure}{.5\textwidth}
  \centering
  \includegraphics[scale=0.43,trim={16pt 5pt 43pt 31pt},clip]{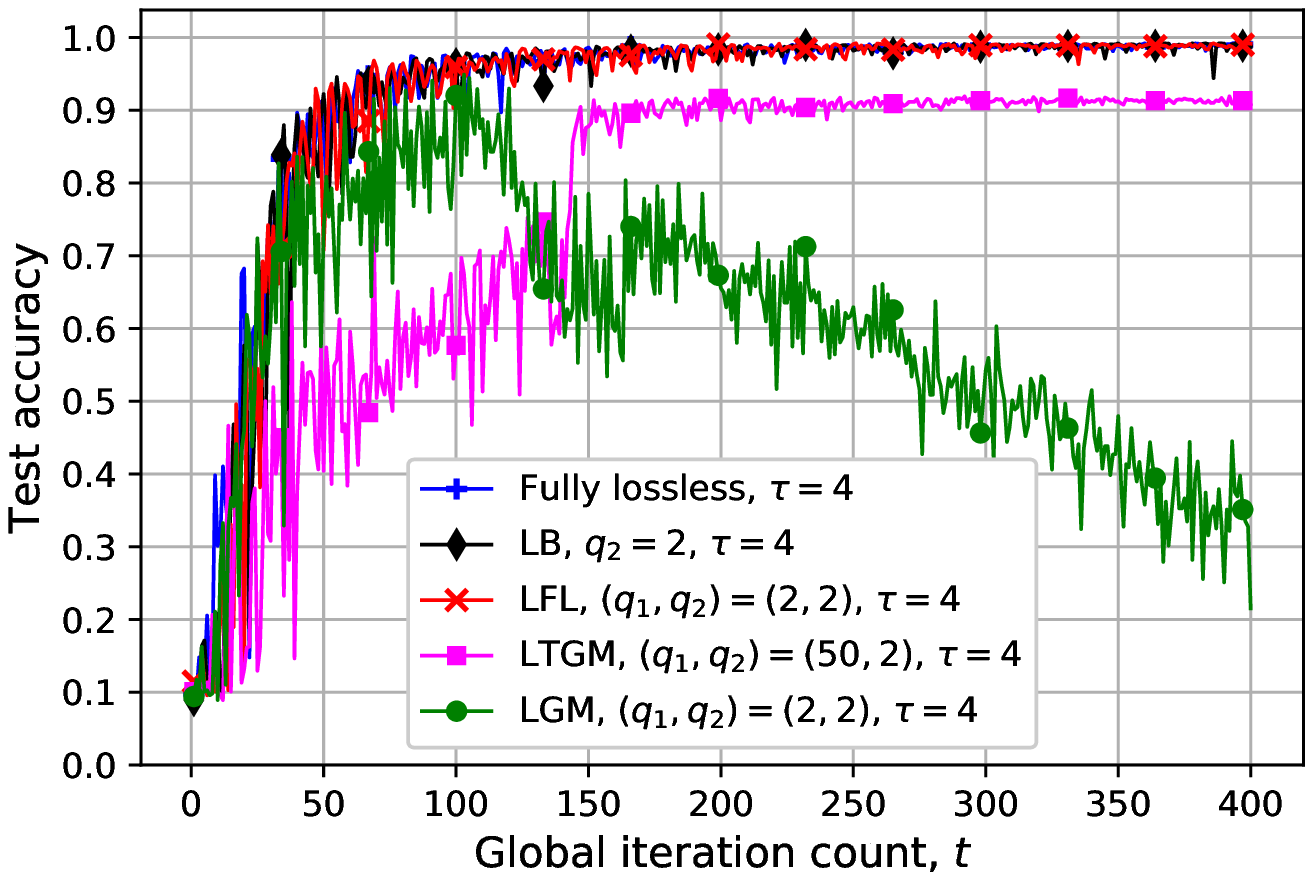}\vspace{0cm}
  \caption{MNIST with non-iid data}
  \label{Fig_MNIST}
\end{subfigure}%
\begin{subfigure}{.5\textwidth}
  \centering
  \includegraphics[scale=0.43,trim={16pt 5pt 43pt 31pt},clip]{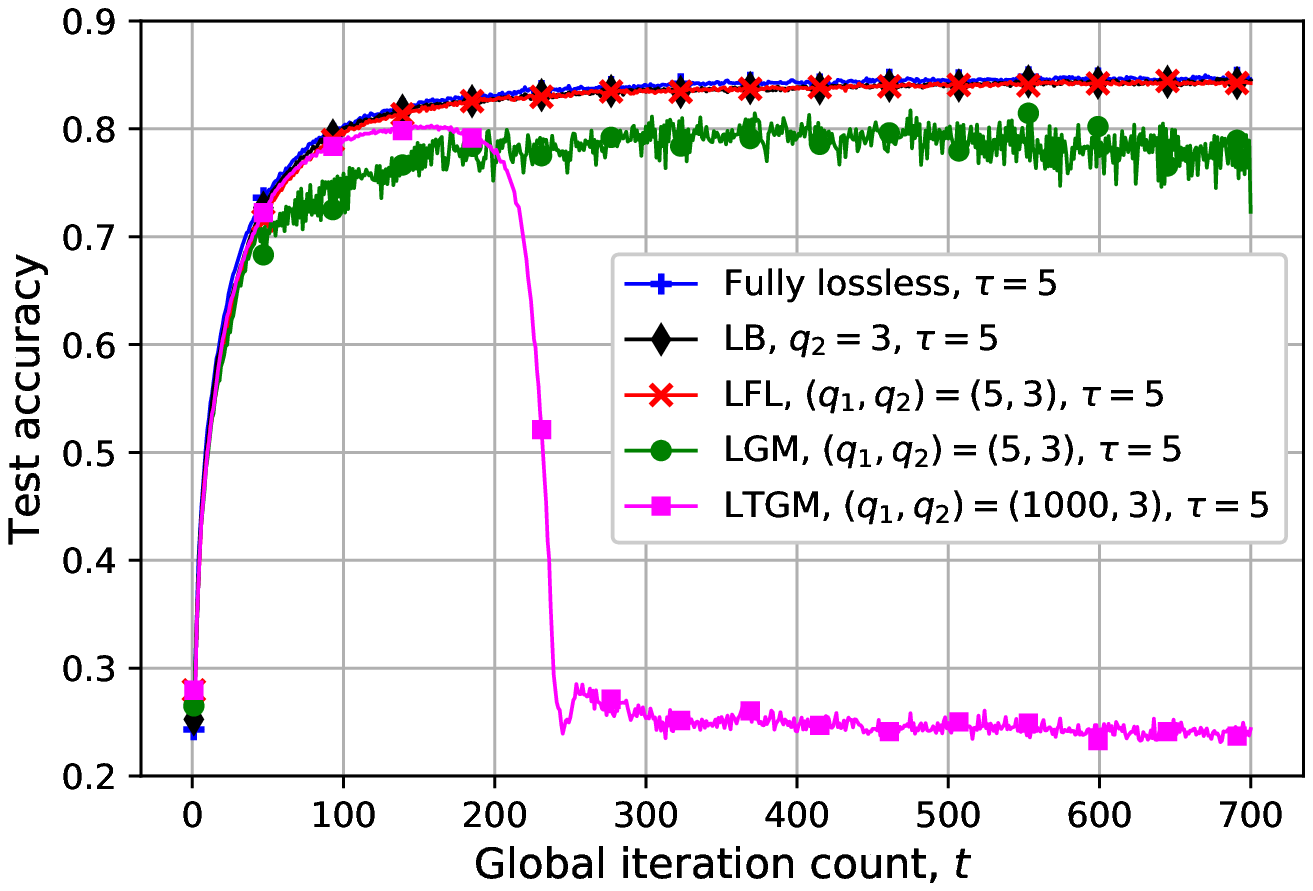}\vspace{0cm}
  \caption{CIFAR-10 with iid data}
  \label{Fig_CIFAR10}
\end{subfigure}\vspace{-.15cm}
\caption{Test accuracy using MNIST and CIFAR-10 for training with local mini-batch size $\left| \xi_m^i (t) \right| = 500$ and $\left| \xi_m^i (t) \right| = 250$, respectively.}
\label{Fig_IID_nonIID_Diff_q_nonIID}
\end{figure}


In Figure \ref{Fig_IID_nonIID_Diff_q_nonIID} we illustrate the performance of different approaches for non-iid and iid scenarios using MNIST and CIFAR-10, respectively, for training with $M=40$ devices. 
Figure \ref{Fig_MNIST} demonstrates test accuracy of different approaches for non-iid data using MNIST with local mini-batch size $\left| \xi_m^i (t) \right| = 500$ and number of local iterations $\tau = 4$. 
We set $q_2=2$ for all the approaches where the devices perform quantization, and $q_1=2$ for the LFL and LGM schemes. 
We observe that the proposed LFL algorithm with $(q_1, q_2) = (2,2)$ performs as good as the fully lossless and LB approaches, despite a factor of $12.77$ savings in the number of bits that need to be broadcast compared to the LB approach.  
This illustrates the efficiency of the LFL algorithm for the iid scenario providing significant communication cost savings without any visible performance degradation.
On the other hand, the performance of the LGM algorithm drops after an intermediate number of training iterations, which shows that the quantization level $q_1=2$ does not provide the devices with an accurate estimate of the global model to rely on for local training.
This is particularly more harmful in later iterations as the algorithm approaches the optimal point where a finer estimate of the global model is needed for training. 
We highlight that the proposed LFL algorithm resolves this deficiency with the LGM algorithm through quantizating the global model update rather than the global model providing a more accurate estimate of the global model to the devices even with a relatively small quantization level $q_1=2$.
Throughout our experiments, we found that the random linear transform with the LTGM scheme is not highly efficient in providing a transformed vector with a relatively small peak-to-average ratio, and the quantization level $q_1$ should be relatively large to guarantee that the algorithm succeeds in learning.  
Therefore, we set $q_1=50$ for the LTGM scheme, which is a relatively large quantization value.
The advantage of the proposed LFL algorithm over the LTGM and LGM algorithms for the non-iid scenario can be clearly seen in the figure.

A similar observation is made in Figure \ref{Fig_CIFAR10} illustrating the perforance of different approaches for iid data using CIFAR-10 with local mini-batch size $\left| \xi_m^i (t) \right| = 250$ and number of local iterations $\tau = 5$.
The the LFL algorithm with $(q_1, q_2) = (5,3)$ provides $\times 9.2$ smaller communication load compared to LB with $q_2 = 3$ without any visible performance degradation with respect to the fully lossless and LB approaches.
It also significantly outperforms the LGM algorithm with $(q_1, q_2) = (5,3)$, which shows the advantage of quantizing the global model update rather than the global model
for iid data. 
We also observe that the accuracy level of the LTGM algorithm drops significantly after around 200 global iterations even for a large quantization level $q_1=1000$, which shows the deficiency of the linear transform to provide a relatively small peak-to-average ratio for the transformed vector.

\vspace{-.25cm}
\section{Conclusion}\label{SecConc}
\vspace{-.15cm}
FL is demanding in terms of bandwidth, particularly when deep networks with huge numbers of parameters are trained across a large number of devices. 
Communication is typically the major bottleneck, since it involves iterative transmission over a bandwidth-limited wireless medium between the PS and a massive number of devices at the edge.
With the goal of reducing the communication cost, we have studied FL with lossy broadcasting, where, in contrast to most of the existing work in the literature, the PS broadcasts a compressed version of the global model to the devices. 
We have considered broadcasting quantized global model updates from the PS, which can be used to estimate the current global model at the devices for local SGD iterations.
The PS aggregates the quantized local model updates from the devices, according to which it updates the global model. 
We have derived convergence guarantees for the proposed LFL algorithm to analyze the impact of lossy broadcasting on the FL performance assuming accurate local model updates at the PS.  
Numerical experiments have shown the efficiency of the proposed LFL algorithm in providing an accurate estimate of the global model to the devices, where it performs as good as the fully lossless and LB approaches for both non-iid and iid data despite the significant reduction in the communication load. 
It also significantly outperforms the LTGM \cite{ExpandingRedClientResFL} and LGM \cite{DoubleSqueezeTangConf} algorithms studying compression in the PS-to-device direction thanks to quantizing the global model update rather than the global model at the PS.

\medskip

\bibliographystyle{unsrt}
\bibliography{nips_2020}

\appendix

\newpage 

\section{Stochastic quantization}\label{AppStochQuan}

Given $\boldsymbol{x} \in \mathbb{R}^d$, with the $i$-th entry denoted by $x_i$, we define
\begin{subequations}
\begin{align}
x_{\rm{max}}& \triangleq \max \left\{ \left| \boldsymbol{x} \right| \right\},\\
x_{\rm{min}}& \triangleq \min \left\{ \left| \boldsymbol{x} \right| \right\}.
\end{align}
\end{subequations}
Given a quantization level $q \ge 1$, we have
\begin{subequations}\label{QSGD}
\begin{align}
{Q} \left({x}_i, q\right) \triangleq {\rm{sign}} \left( {x}_i \right) \cdot \Big( x_{\rm{min}} + \left(x_{\rm{max}} - x_{\rm{min}} \right) \cdot  \varphi \Big( \frac{\left| {x}_i \right| - x_{\rm{min}}}{x_{\rm{max}} -x_{\rm{min}}}, q \Big) \Big), \quad \mbox{for $i \in [d]$},    
\end{align}
where $\varphi(\cdot, \cdot)$ is a quantization function defined in the following. For $0 \le x \le 1$ and $q \ge 1$, let $l \in \{ 0, 1, \dots, q-1 \}$ be an integer such that $x \in [l/q, (l+1)/q)$. We then define
\begin{align}\label{DefQ}
\varphi \left( x, q \right) \triangleq \begin{cases} 
l / q, & \mbox{with probability $1 - \left( x q - l\right)$},\\
(l+1) / q, & \mbox{with probability $x q - l$}.
\end{cases}
\end{align}
\end{subequations}
We define 
\begin{align}\label{QVecDefinition}
\boldsymbol{Q}(\boldsymbol{x}, q) \triangleq [Q({x}_1, q), \cdots, Q({x}_d, q) ]^T,    
\end{align}
and we highlight that it is represented by 
\begin{align}
R_{\rm{Q}} = 64 + d \left( 1 + \log_2(q+1) \right) \mbox{ bits,}    
\end{align}
where $64$ bits are used to represent $x_{\rm{max}}$ and $x_{\rm{min}}$, $d$ bits are used for ${\rm{sign}} (x_i)$, $\forall i \in [d]$, and $d \log_2(q+1)$ bits represent $\varphi \left( (\left| {x}_i \right|-x_{\rm{min}})/ (x_{\rm{max}} - x_{\rm{min}}), q \right)$, $\forall i \in [d]$. We note that we have modified the QSGD scheme proposed in \cite{DCAlistarhQSGD} by normalizing the
entries of vector $\boldsymbol{x}$ with $x_{\rm{max}} - x_{\rm{min}}$ rather than $\left\| \boldsymbol{x} \right\|_2$.

\section{Proof of Lemma \ref{LemDigQuanVar}}\label{AppProofLemma}

Given $\varphi \left( x, q \right)$ in (\ref{DefQ}), we have
\begin{align}\label{ProofLemma1_Eq1}
\mathbb{E}_{\varphi} \left[  \varphi(x, q) \right] & =  \Big(\frac{l}{q} \Big) \left( 1 + l - xq \right) + \Big(\frac{l+1}{q} \Big) \left( xq - l \right) = x.
\end{align}
Also, we have
\newcommand\ainequalle{\mathrel{\overset{\makebox[0pt]{\mbox{\normalfont\tiny\sffamily (a)}}}{\le}}}
\begin{align}\label{ProofLemma1_Eq2}
\mathbb{E}_{\varphi} \left[ \varphi^2(x, q) \right] & = \Big(\frac{l}{q} \Big)^2 \left( 1 + l - xq \right) + \Big(\frac{l+1}{q} \Big)^2 \left( xq - l \right)  = \frac{1}{q^2} \left( -l^2 + 2lxq + xq - l \right) \nonumber\\
&  = x^2 + \frac{1}{q^2} \left( xq - l \right) \left( 1 - xq + l \right) \ainequalle x^2 + \frac{1}{4 q^2},
\end{align}
where (a) follows since $\left( xq - l \right) \left( 1 - xq + l \right) \le 1/4$.
According to (\ref{ProofLemma1_Eq1}), (\ref{ProofLemma1_Eq2}) and the definition of $\boldsymbol{Q}(\boldsymbol{x}, q)$ given in (\ref{QVecDefinition}), it follows that
\newcommand\binequalle{\mathrel{\overset{\makebox[0pt]{\mbox{\normalfont\tiny\sffamily (b)}}}{\le}}}
\newcommand\onlyhereineq{\mathrel{\overset{\makebox[0pt]{\mbox{\normalfont\tiny\sffamily (c)}}}{\le}}}
\begin{subequations}\label{QuantizationBounds}
\begin{align}\label{Quantization_Expectation}
&\mathbb{E}_{\varphi} \left[  \boldsymbol{Q}(\boldsymbol{x}, q) \right] = \boldsymbol{x}, \\
&\mathbb{E}_{\varphi} \big[  \left\| \boldsymbol{Q}(\boldsymbol{x}, q) \right\|^2_2 \big] =  \sum\nolimits_{i=1}^{d} \mathbb{E}_{\varphi} \big[  \left| Q({x}_i, q) \right|^2_2 \big] = \left(x_{\rm{max}} - x_{\rm{min}} \right)^2 \sum\nolimits_{i=1}^{d} \mathbb{E}_{\varphi} \Big[ \varphi^2 \Big( \frac{\left| {x}_i \right| - x_{\rm{min}}}{x_{\rm{max}} -x_{\rm{min}}}, q \Big) \Big]  \nonumber\\
& \quad \;\; \quad + d x^2_{\rm{min}} + 2 x_{\rm{min}} (x_{\rm{max}} -x_{\rm{min}}) \sum\nolimits_{i=1}^{d} \mathbb{E}_{\varphi} \Big[ \varphi \Big( \frac{\left| {x}_i \right| - x_{\rm{min}}}{x_{\rm{max}} -x_{\rm{min}}}, q \Big) \Big] \nonumber\\
& \quad \;\; \binequalle \left(x_{\rm{max}} - x_{\rm{min}}\right)^2 \sum\nolimits_{i=1}^{d} \left( \Big( \frac{\left| {x}_i \right| - x_{\rm{min}}}{x_{\rm{max}} - x_{\rm{min}}} \Big)^2 + \frac{1}{4 q^2} \right)+ d x^2_{\rm{min}} + 2 x_{\rm{min}} \sum\nolimits_{i=1}^{d}( \left| {x}_i \right| - x_{\rm{min}}) \nonumber\\
& \quad \;\; = \left\| \boldsymbol{x} \right\|_2^2 + d \frac{\left(x_{\rm{max}} - x_{\rm{min}}\right)^2}{4q^2} \onlyhereineq \left\| \boldsymbol{x} \right\|_2^2 + \frac{\varepsilon d \left\| \boldsymbol{x} \right\|_2^2}{4 q^2},\label{Quantization_SquareExpectation}
\end{align}
\end{subequations}
where (b) follows from (\ref{ProofLemma1_Eq1}) and (\ref{ProofLemma1_Eq2}), and (c) follows since $\varepsilon = \left(x_{\rm{max}} - x_{\rm{min}}\right)^2 /  \left\| \boldsymbol{x} \right\|_2^2$.

\section{Proof of Theorem \ref{Theoremtheta_thetastarKequalM}}\label{AppProofTheorem}
We have 
\begin{align}\label{A_AppFDP_1}
\mathbb{E} \left[ \left\| \boldsymbol{\theta} (t+1) - {\boldsymbol{\theta}}^* \right\|_2^2 \right] = & \mathbb{E} \left[ \left\| \wh{\boldsymbol{\theta}} (t) - {\boldsymbol{\theta}}^* \right\|_2^2 \right] + \mathbb{E} \bigg[ \left\| \sum\nolimits_{m =1}^M \frac{B_m}{B} \Delta \boldsymbol{\theta}_m (t) \right\|_2^2 \bigg] \nonumber\\
& + 2 \mathbb{E} \left[ \langle \wh{\boldsymbol{\theta}} (t) - {\boldsymbol{\theta}}^* , \sum\nolimits_{m =1}^M \frac{B_m}{B} \Delta \boldsymbol{\theta}_m (t) \rangle \right].  
\end{align}
In the following, we bound the last two terms on the right hand side (RHS) of (\ref{A_AppFDP_1}).
From the convexity of $\left\| \cdot \right\|_2^2$, it follows that
\begin{align}\label{AppLemmaTemr_2_Eq_1_2}
&\mathbb{E} \bigg[ \left\| \sum\nolimits_{m =1}^M \frac{B_m}{B} \Delta \boldsymbol{\theta}_m (t) \right\|_2^2 \bigg] \le  \sum\nolimits_{m =1}^{M} \frac{B_m}{B} \mathbb{E} \left[ \left\| \Delta \boldsymbol{\theta}_m (t) \right\|_2^2 \right]  \nonumber\\
& \qquad \qquad \quad = \eta^2(t) \sum\nolimits_{m =1}^{M} \frac{B_m}{B} \mathbb{E} \left[ \left\| \sum\nolimits_{i=1}^{\tau} \nabla F_m \left( \boldsymbol{\theta}_m^i (t), \xi_m^i (t) \right) \right\|_2^2 \right] \nonumber\\
& \qquad \qquad \quad  \le \eta^2(t) \tau \sum\nolimits_{m =1}^{M} \sum\nolimits_{i=1}^{\tau} \frac{B_m}{B} \mathbb{E} \left[ \left\| \nabla F_m \left( \boldsymbol{\theta}_m^i (t), \xi_m^i (t) \right) \right\|_2^2 \right] \ainequalle \eta^2(t) \tau^2 G^2,
\end{align}
where (a) follows from Assumption \ref{AssumpBoundedVarGradient}.

We rewrite the third term on the RHS of (\ref{A_AppFDP_1}) as follows:
\begin{align}\label{A_AppLemmaTemr_2_Eq_2}
&2 \mathbb{E} \left[ \langle \wh{\boldsymbol{\theta}} (t) - {\boldsymbol{\theta}}^* ,  \sum\nolimits_{m =1}^{M} \frac{B_m}{B} \Delta \boldsymbol{\theta}_m (t) \rangle \right] \nonumber\\
& \qquad = 2 \eta(t) \sum\nolimits_{m=1}^{M} \frac{B_m}{B} \mathbb{E} \left[ \langle {\boldsymbol{\theta}}^* - \wh{\boldsymbol{\theta}} (t) , \sum\nolimits_{i=1}^{\tau} \nabla F_m \left( \boldsymbol{\theta}_m^i (t), \xi_m^i (t) \right) \rangle \right] \nonumber\\
& \qquad = 2 \eta(t) \sum\nolimits_{m=1}^{M} \frac{B_m}{B} \mathbb{E} \left[ \langle {\boldsymbol{\theta}}^* - \wh{\boldsymbol{\theta}} (t) , \nabla F_m \big( \wh{\boldsymbol{\theta}} (t), \xi_m^1 (t) \big) \rangle \right] \nonumber\\
& \qquad \quad + 2 \eta(t) \sum\nolimits_{m=1}^{M} \frac{B_m}{B} \mathbb{E} \left[ \langle {\boldsymbol{\theta}}^* - \wh{\boldsymbol{\theta}} (t) , \sum\nolimits_{i=2}^{\tau} \nabla F_m \left( \boldsymbol{\theta}_m^i (t), \xi_m^i (t) \right) \rangle \right]. 
\end{align}
We have
\newcommand\thirdequal{\mathrel{\overset{\makebox[0pt]{\mbox{\normalfont\tiny\sffamily (a)}}}{=}}}
\newcommand\thirdinequal{\mathrel{\overset{\makebox[0pt]{\mbox{\normalfont\tiny\sffamily (b)}}}{\le}}}
\newcommand\cinequal{\mathrel{\overset{\makebox[0pt]{\mbox{\normalfont\tiny\sffamily (c)}}}{\le}}}
\begin{align}\label{AppLemmaTemr_2_Eq_3}
& 2 \eta(t) \sum\nolimits_{m=1}^{M} \frac{B_m}{B} \mathbb{E} \left[ \langle {\boldsymbol{\theta}}^* - \wh{\boldsymbol{\theta}} (t) , \nabla F_m \big( \wh{\boldsymbol{\theta}} (t), \xi_m^1 (t) \big) \rangle \right] \nonumber\\
& \qquad \qquad \qquad \qquad \thirdequal 2 \eta(t) \sum\nolimits_{m=1}^{M} \frac{B_m}{B} \mathbb{E} \left[ \langle {\boldsymbol{\theta}}^* - \wh{\boldsymbol{\theta}} (t) , \nabla F_m \big( \wh{\boldsymbol{\theta}} (t) \big) \rangle \right] \nonumber\\
& \qquad \qquad \qquad \qquad \thirdinequal 2 \eta(t) \sum\nolimits_{m=1}^{M} \frac{B_m}{B} \mathbb{E} \left[ F_m (\boldsymbol{\theta}^*) - F_m \big( \wh{\boldsymbol{\theta}} (t) \big) - \frac{\mu}{2} \left\| \wh{\boldsymbol{\theta}} (t) - {\boldsymbol{\theta}}^* \right\|_2^2 \right] \nonumber\\
& \qquad \qquad \qquad \qquad = 2 \eta(t) \Big( F^* - \mathbb{E} \left[ F\big( \wh{\boldsymbol{\theta}} (t) \big) \right] - \frac{\mu}{2} \mathbb{E} \Big[ \left\| \wh{\boldsymbol{\theta}} (t) - {\boldsymbol{\theta}}^* \right\|_2^2 \Big] \Big)\nonumber\\
& \qquad \qquad \qquad \qquad \cinequal -\mu \eta(t) \mathbb{E} \Big[ \left\| \wh{\boldsymbol{\theta}} (t) - {\boldsymbol{\theta}}^* \right\|_2^2 \Big],
\end{align}
where (a) follows since $\mathbb{E}_{\xi} \left[ \nabla F_m \left( \boldsymbol{\theta}_m^i (t), \xi_m^i (t) \right) \right] = \nabla F_m \left( \boldsymbol{\theta}_m^i (t)  \right)$, $\forall i, m$, (b) follows from Assumption \ref{AssumpStrongConvexLoss}, and (c) follows since $F^* \le F\big( \wh{\boldsymbol{\theta}} (t) \big)$, $\forall t$.

\begin{lemma}\label{App_lemma_term2}
For $0 < \eta(t) \le 1$, we have
\begin{align}\label{EQ_LemmaTermE}
& 2 \eta(t) \sum\nolimits_{m=1}^{M} \frac{B_m}{B} \mathbb{E} \left[ \langle {\boldsymbol{\theta}}^* - \wh{\boldsymbol{\theta}} (t) , \sum\nolimits_{i=2}^{\tau} \nabla F_m \left( \boldsymbol{\theta}_m^i (t), \xi_m^i (t) \right) \rangle \right]  \nonumber\\
& \qquad   \le - \mu \eta(t) (1 - \eta(t)) (\tau - 1) \mathbb{E} \left[ \left\| \wh{\boldsymbol{\theta}} (t) - \boldsymbol{\theta}^* \right\|_2^2 \right] + \eta^2 (t) \left( \tau - 1 \right) G^2 \nonumber\\  
& \qquad  \quad + (1+ \mu (1 - \eta(t))) \eta^2(t) G^2 \frac{\tau (\tau-1)(2\tau-1)}{6} + 2 \eta(t) (\tau - 1) \Gamma . 
\end{align}
\end{lemma}
\begin{proof}
See Appendix \ref{App_Proof_lemma}. 
\end{proof}

By substituting (\ref{AppLemmaTemr_2_Eq_3}) and (\ref{EQ_LemmaTermE}) in (\ref{A_AppLemmaTemr_2_Eq_2}), it follows that
\begin{align}\label{EQ_LemmaTermE_combined}
& 2 \mathbb{E} \left[ \langle \wh{\boldsymbol{\theta}} (t) - {\boldsymbol{\theta}}^* ,  \sum\nolimits_{m =1}^{M} \frac{B_m}{B} \Delta \boldsymbol{\theta}_m (t) \rangle \right]  \nonumber\\
& \qquad   \le - \mu \eta(t) (\tau - \eta(t) (\tau - 1)) \mathbb{E} \left[ \left\| \wh{\boldsymbol{\theta}} (t) - \boldsymbol{\theta}^* \right\|_2^2 \right] + \eta^2 (t) \left( \tau - 1 \right) G^2 \nonumber\\  
& \qquad  \quad + (1+ \mu (1 - \eta(t))) \eta^2(t) G^2 \frac{\tau (\tau-1)(2\tau-1)}{6} + 2 \eta(t) (\tau - 1) \Gamma, 
\end{align}
which, together with the inequality in (\ref{AppLemmaTemr_2_Eq_1_2}), leads to the following upper bound on $\mathbb{E} \left[ \left\| \boldsymbol{\theta} (t+1) - {\boldsymbol{\theta}}^* \right\|_2^2 \right]$, when substituted into (\ref{A_AppFDP_1}):
\begin{align}\label{APP1_Final_0}
\mathbb{E} \left[ \left\| \boldsymbol{\theta} (t+1) - {\boldsymbol{\theta}}^* \right\|_2^2 \right] \le & \left(1 - \mu \eta(t) (\tau - \eta(t) (\tau - 1)) \right) \mathbb{E} \left[ \left\| \wh{\boldsymbol{\theta}} (t) - \boldsymbol{\theta}^* \right\|_2^2 \right] + \eta^2 (t) \left( \tau^2+ \tau - 1 \right) G^2  \nonumber\\  
& + (1+ \mu (1 - \eta(t))) \eta^2(t) G^2 \frac{\tau (\tau-1)(2\tau-1)}{6} + 2 \eta(t) (\tau - 1) \Gamma.
\end{align}

\begin{lemma}\label{Lemma_ThetaHatTheta}
For $\wh{\boldsymbol{\theta}} (t)$ given in (\ref{ThetaHat_t}), we have
\begin{align}\label{EQ_Lemma_ThetaHatTheta}
\mathbb{E} \left[ \left\| \wh{\boldsymbol{\theta}} (t) - {\boldsymbol{\theta}}^* \right\|_2^2 \right] \le \mathbb{E} \left[ \left\| {\boldsymbol{\theta}} (t) - {\boldsymbol{\theta}}^* \right\|_2^2 \right] + \Big(\frac{\eta(t-1) \tau G}{2q_1(t)}\Big)^2 \varepsilon d. 
\end{align}
for some $0 \le \varepsilon \le 1$.
\end{lemma}
\begin{proof}
See Appendix \ref{App_Proof_lemma_2}. 
\end{proof}

According to Lemma \ref{Lemma_ThetaHatTheta}, the inequality in (\ref{EQ_Lemma_ThetaHatTheta}) can be rewritten as follows:
\begin{align}\label{APP1_Final}
\mathbb{E} &\left[ \left\| \boldsymbol{\theta} (t+1) - {\boldsymbol{\theta}}^* \right\|_2^2 \right] \le  \left(1 - \mu \eta(t) (\tau - \eta(t) (\tau - 1)) \right) \mathbb{E} \left[ \left\| {\boldsymbol{\theta}} (t) - \boldsymbol{\theta}^* \right\|_2^2 \right]   \nonumber\\  
& \qquad \qquad + \left( 1 - \mu \eta (t) \left( \tau - \eta(t) (\tau - 1) \right) \right) \Big(\frac{\eta(t-1) \tau G}{2q_1(t)}\Big)^2 \varepsilon d + \eta^2 (t) \left( \tau^2+ \tau - 1 \right) G^2 \nonumber\\
& \qquad \qquad + (1+ \mu (1 - \eta(t))) \eta^2(t) G^2 \frac{\tau (\tau-1)(2\tau-1)}{6} + 2 \eta(t) (\tau - 1) \Gamma.    
\end{align}
Theorem \ref{Theoremtheta_thetastarKequalM} follows from the inequality in (\ref{APP1_Final}) having $0 < \eta(t) \le \min \left\{ 1, \frac{1}{\mu \tau} \right\}$, $\forall t$.

\section{Proof of Lemma \ref{App_lemma_term2}}\label{App_Proof_lemma}
We have 
\begin{align}\label{EQ_LemmaTermE_app_1}
& 2 \eta(t) \sum\nolimits_{m=1}^{M} \frac{B_m}{B} \sum\nolimits_{i=2}^{\tau} \mathbb{E} \left[ \langle \boldsymbol{\theta}^* - \wh{\boldsymbol{\theta}} (t) , \nabla F_m \left( \boldsymbol{\theta}_m^i (t), \xi_m^i (t) \right) \rangle \right] \nonumber\\
& \quad = 2 \eta(t) \sum\nolimits_{m=1}^{M} \frac{B_m}{B} \sum\nolimits_{i=2}^{\tau} \mathbb{E} \left[ \langle \boldsymbol{\theta}_m^i (t) - \wh{\boldsymbol{\theta}} (t) , \nabla F_m \left( \boldsymbol{\theta}_m^i (t), \xi_m^i (t) \right) \rangle \right]\nonumber\\
& \quad \;\; \;\;+ 2 \eta(t) \sum\nolimits_{m=1}^{M} \frac{B_m}{B} \sum\nolimits_{i=2}^{\tau} \mathbb{E} \left[ \langle \boldsymbol{\theta}^* - \boldsymbol{\theta}_m^i (t),  \nabla F_m \left( \boldsymbol{\theta}_m^i (t), \xi_m^i (t) \right) \rangle \right].  
\end{align}
We first bound the first term on the RHS of (\ref{EQ_LemmaTermE_app_1}). We have
\begin{align}\label{AppLemmaTemr_1_D_Eq_1}
& 2 \eta(t) \sum\nolimits_{m=1}^{M} \frac{B_m}{B} \sum\nolimits_{i=2}^{\tau} \mathbb{E} \left[ \langle \boldsymbol{\theta}_m^i (t) - \wh{\boldsymbol{\theta}} (t) , \nabla F_m \left( \boldsymbol{\theta}_m^i (t), \xi_m^i (t) \right) \rangle \right]\nonumber \\
& \; \; \; \quad \le \eta(t) \sum\nolimits_{m=1}^{M} \frac{B_m}{B} \sum\nolimits_{i=2}^{\tau} \mathbb{E} \bigg[ \frac{1}{\eta(t)} \left\| \boldsymbol{\theta}_m^i (t) - \wh{\boldsymbol{\theta}} (t) \right\|_2^2 + \eta(t) \left\| \nabla F_m \left( \boldsymbol{\theta}_m^i (t), \xi_m^i (t) \right) \right\|_2^2 \bigg]\nonumber\\
& \; \; \; \quad \ainequalle \sum\nolimits_{m=1}^{M} \frac{B_m}{B} \sum\nolimits_{i=2}^{\tau} \mathbb{E} \left[ \left\| \boldsymbol{\theta}_m^i (t) - \wh{\boldsymbol{\theta}} (t) \right\|_2^2 \right] + \eta^2 (t) \left( \tau - 1 \right) G^2,
\end{align}
where (a) follows from Assumption \ref{AssumpBoundedVarGradient}.
We have
\begin{align}\label{AppLemmaTemr_2_Eq_7}
&\sum\limits_{m=1}^{M} \frac{B_m}{B} \sum\limits_{i=2}^{\tau} \mathbb{E} \left[ \left\| \boldsymbol{\theta}_m^i (t) - \wh{\boldsymbol{\theta}} (t) \right\|_2^2 \right] \nonumber\\
& = \eta^2(t) \sum\limits_{m=1}^{M} \frac{B_m}{B} \sum\limits_{i=2}^{\tau} \mathbb{E} \left[ \left\| \sum\nolimits_{j=1}^{i} \nabla F_m \left( \boldsymbol{\theta}_m^j (t), \xi_m^j (t) \right) \right\|_2^2 \right] \binequalle \eta^2(t) G^2 \frac{\tau (\tau-1)(2\tau-1)}{6}, 
\end{align}
where (b) follows from the convexity of $\left\| \cdot \right\|_2^2$ and Assumption \ref{AssumpBoundedVarGradient}.
Plugging (\ref{AppLemmaTemr_2_Eq_7}) into (\ref{AppLemmaTemr_1_D_Eq_1}) yields  
\begin{align}\label{AppLemmaTemr_1_D_Eq_2}
& 2 \eta(t) \sum\nolimits_{m=1}^{M} \frac{B_m}{B} \sum\nolimits_{i=2}^{\tau} \mathbb{E} \left[ \langle \boldsymbol{\theta}_m^i (t) - \wh{\boldsymbol{\theta}} (t) , \nabla F_m \left( \boldsymbol{\theta}_m^i (t), \xi_m^i (t) \right) \rangle \right]\nonumber \\
& \; \; \; \qquad \qquad \qquad \qquad \qquad \qquad \le \eta^2(t) G^2 \frac{\tau (\tau-1)(2\tau-1)}{6} + \eta^2 (t) \left( \tau - 1 \right) G^2.
\end{align}
For the second term on the RHS of (\ref{EQ_LemmaTermE_app_1}), we have
\newcommand\cinequalle{\mathrel{\overset{\makebox[0pt]{\mbox{\normalfont\tiny\sffamily (c)}}}{\le}}}
\begin{align}\label{EQ_LemmaTermE_D_app_1}
& 2 \eta(t) \sum\nolimits_{m=1}^{M} \frac{B_m}{B} \sum\nolimits_{i=2}^{\tau} \mathbb{E} \left[ \langle \boldsymbol{\theta}^* - \boldsymbol{\theta}_m^i (t) , \nabla F_m \left( \boldsymbol{\theta}_m^i (t), \xi_m^i (t) \right) \rangle \right] \nonumber\\
& \thirdequal 2 \eta(t) \sum\nolimits_{m=1}^{M} \frac{B_m}{B} \sum\nolimits_{i=2}^{\tau} \mathbb{E} \left[ \langle \boldsymbol{\theta}^* - \boldsymbol{\theta}_m^i (t) , \nabla F_m \left( \boldsymbol{\theta}_m^i (t) \right) \rangle \right]\nonumber\\
& \thirdinequal 2 \eta(t) \sum\nolimits_{m=1}^{M} \frac{B_m}{B} \sum\nolimits_{i=2}^{\tau} \mathbb{E} \left[ F_m (\boldsymbol{\theta}^*) - F_m(\boldsymbol{\theta}_m^i (t)) - \frac{\mu}{2} \left\| \boldsymbol{\theta}_m^i (t) - {\boldsymbol{\theta}}^* \right\|_2^2 \right]\nonumber\\
& = 2 \eta(t) \sum\nolimits_{m=1}^{M} \frac{B_m}{B} \sum\nolimits_{i=2}^{\tau} \mathbb{E} \left[ F_m (\boldsymbol{\theta}^*) - F_m^* + F_m^* - F_m(\boldsymbol{\theta}_m^i (t)) - \frac{\mu}{2} \left\| \boldsymbol{\theta}_m^i (t) - {\boldsymbol{\theta}}^* \right\|_2^2 \right]\nonumber\\
& = 2 \eta(t) (\tau - 1) \Gamma +  2 \eta(t) \sum\nolimits_{m=1}^{M} \frac{B_m}{B} \sum\nolimits_{i=2}^{\tau} \left( F_m^* - \mathbb{E} \left[ F_m({\boldsymbol{\theta}}_m^i (t)) \right] \right) \nonumber\\
& \quad - \mu \eta(t) \sum\nolimits_{m=1}^{M} \frac{B_m}{B} \sum\nolimits_{i=2}^{\tau} \mathbb{E} \left[ \left\| \boldsymbol{\theta}_m^i (t) - {\boldsymbol{\theta}}^* \right\|_2^2 \right] \nonumber\\
& \cinequalle 2 \eta(t) (\tau - 1) \Gamma - \mu \eta(t) \sum\nolimits_{m=1}^{M} \frac{B_m}{B} \sum\nolimits_{i=2}^{\tau} \mathbb{E} \left[ \left\| \boldsymbol{\theta}_m^i (t) - {\boldsymbol{\theta}}^* \right\|_2^2 \right],
\end{align}
where (a) follows since $\mathbb{E}_{\xi} \left[ \nabla F_m \left( \boldsymbol{\theta} (t), \xi_m^i (t) \right) \right] = \nabla F_m \left( \boldsymbol{\theta} (t)  \right)$, $\forall i, m, t$; (b) follows from Assumption \ref{AssumpStrongConvexLoss}; and (c) follows since $F^*_m \le F_m({\boldsymbol{\theta}}_m^i (t))$, $\forall m$.
We have
\begin{align}\label{EQ_LemmaTermE_app_2}
& - \left\| \boldsymbol{\theta}_m^i (t) - {\boldsymbol{\theta}}^* \right\|_2^2 = - \left\| \boldsymbol{\theta}_m^i (t) - \widehat{\boldsymbol{\theta}} (t) \right\|_2^2 - \left\| \widehat{\boldsymbol{\theta}} (t) - {\boldsymbol{\theta}}^* \right\|_2^2 - 2 \langle \boldsymbol{\theta}_m^i (t) - \widehat{\boldsymbol{\theta}} (t) , \widehat{\boldsymbol{\theta}} (t) - {\boldsymbol{\theta}}^* \rangle \nonumber\\
& \qquad \quad \ainequalle - \left\| \boldsymbol{\theta}_m^i (t) - \widehat{\boldsymbol{\theta}} (t) \right\|_2^2 - \left\| \widehat{\boldsymbol{\theta}} (t) - {\boldsymbol{\theta}}^* \right\|_2^2 + \frac{1}{\eta(t)} \left\| \boldsymbol{\theta}_m^i (t) - \widehat{\boldsymbol{\theta}} (t) \right\|_2^2 + \eta(t) \left\| \widehat{\boldsymbol{\theta}} (t) - {\boldsymbol{\theta}}^* \right\|_2^2 \nonumber\\
&\qquad \quad = - (1 - \eta(t)) \left\| \widehat{\boldsymbol{\theta}} (t) - {\boldsymbol{\theta}}^* \right\|_2^2 + \Big(\frac{1}{\eta(t)} - 1\Big) \left\| \boldsymbol{\theta}_m^i (t) - \widehat{\boldsymbol{\theta}} (t) \right\|_2^2,
\end{align}
where (a) follows from Cauchy-Schwarz inequality. 
Plugging (\ref{EQ_LemmaTermE_app_2}) into (\ref{EQ_LemmaTermE_D_app_1}) yields
\begin{align}\label{EQ_LemmaTermE_app_3}
& \frac{2 \eta(t)}{M} \sum\nolimits_{m=1}^{M} \sum\nolimits_{i=2}^{\tau} \mathbb{E} \left[ \langle \boldsymbol{\theta}^* - \boldsymbol{\theta}_m^i (t) , \nabla F_m \left( \boldsymbol{\theta}_m^i (t), \xi_m^i (t) \right) \rangle \right] \nonumber\\
& \le - \mu \eta(t) (1 - \eta(t)) (\tau - 1)  \left\| \widehat{\boldsymbol{\theta}} (t) - {\boldsymbol{\theta}}^* \right\|_2^2 + \mu (1 - \eta(t)) \eta^2(t) G^2 \frac{\tau (\tau-1)(2\tau-1)}{6} + 2 \eta(t) (\tau - 1) \Gamma,
\end{align}
where we used the inequality in (\ref{AppLemmaTemr_2_Eq_7}) and $\eta(t) \le 1$. Plugging (\ref{AppLemmaTemr_1_D_Eq_2}) and (\ref{EQ_LemmaTermE_app_3}) into (\ref{EQ_LemmaTermE_app_1}) completes the proof of Lemma \ref{App_lemma_term2}.

\section{Proof of Lemma \ref{Lemma_ThetaHatTheta}}\label{App_Proof_lemma_2}
We have 
\begin{align}\label{App_ThetaHat_1}
\mathbb{E} \left[ \left\| \wh{\boldsymbol{\theta}} (t) - {\boldsymbol{\theta}}^* \right\|_2^2 \right] & = \mathbb{E} \left[ \left\| \wh{\boldsymbol{\theta}} (t)  \right\|_2^2 \right] + \mathbb{E} \left[ \left\| {\boldsymbol{\theta}}^* \right\|_2^2 \right] - 2 \mathbb{E} \left[ \langle \wh{\boldsymbol{\theta}} (t) , {\boldsymbol{\theta}}^* \rangle \right] \nonumber\\  
& \thirdequal \mathbb{E} \left[ \left\| \wh{\boldsymbol{\theta}} (t)  \right\|_2^2 \right] + \mathbb{E} \left[ \left\| {\boldsymbol{\theta}}^* \right\|_2^2 \right] - 2 \mathbb{E} \left[ \langle {\boldsymbol{\theta}} (t) , {\boldsymbol{\theta}}^* \rangle \right],
\end{align}
where (a) follows since 
\begin{align}
\mathbb{E} \left[ \wh{\boldsymbol{\theta}} (t)  \right] =  \mathbb{E} \left[ \wh{\boldsymbol{\theta}} (t-1)  \right] + \mathbb{E} \left[ \boldsymbol{Q}\big(\boldsymbol{\theta} (t) - \wh{\boldsymbol{\theta}} (t-1), q_1(t)\big)  \right] = \mathbb{E} \left[ {\boldsymbol{\theta}} (t)  \right],
\end{align}
where the last equality follows from (\ref{Quantization_Expectation}). In the following, we upper bound $\mathbb{E} \left[ \left\| \wh{\boldsymbol{\theta}} (t)  \right\|_2^2 \right]$. We have
\begin{align}\label{App_ThetaHatWithA}
\mathbb{E} \left[ \left\| \wh{\boldsymbol{\theta}} (t)  \right\|_2^2 \right] =& \mathbb{E} \left[ \left\| \wh{\boldsymbol{\theta}} (t-1)  \right\|_2^2 \right] + \mathbb{E} \left[ \left\| \boldsymbol{Q}\big(\boldsymbol{\theta} (t) - \wh{\boldsymbol{\theta}} (t-1), q_1(t)\big)  \right\|_2^2 \right] \nonumber\\
& +  2 \mathbb{E} \left[ \langle \wh{\boldsymbol{\theta}} (t-1) , \boldsymbol{Q}\big(\boldsymbol{\theta} (t) - \wh{\boldsymbol{\theta}} (t-1), q_1(t)\big) \rangle \right] \nonumber\\
\ainequalle & \mathbb{E} \left[ \left\| \wh{\boldsymbol{\theta}} (t-1)  \right\|_2^2 \right] + \mathbb{E} \left[ \left\| \boldsymbol{\theta} (t) - \wh{\boldsymbol{\theta}} (t-1) \right\|_2^2 \right] + \frac{\varepsilon(t) d}{4 q_1^2(t)} \mathbb{E} \left[ \left\| \boldsymbol{\theta} (t) - \wh{\boldsymbol{\theta}} (t-1) \right\|_2^2 \right]\nonumber\\
& + 2 \mathbb{E} \left[ \langle \wh{\boldsymbol{\theta}} (t-1),  \boldsymbol{\theta} (t) - \wh{\boldsymbol{\theta}} (t-1)\big) \rangle \right]\nonumber\\
\binequalle & \mathbb{E} \left[ \left\| {\boldsymbol{\theta}} (t)  \right\|_2^2 \right] + \frac{\varepsilon d}{4 q_1^2(t)} \mathbb{E} \left[ \left\| \boldsymbol{\theta} (t) - \wh{\boldsymbol{\theta}} (t-1) \right\|_2^2 \right],
\end{align}
where (a) follows from (\ref{QuantizationBounds}) for some $0 \le \varepsilon (t) \le 1$ defined as
\begin{align}\label{EpsilontDefin}
\varepsilon (t) \triangleq \frac{\mathbb{E} \left[ \left( \max \left\{ \left| \boldsymbol{\theta} (t) - \wh{\boldsymbol{\theta}} (t-1) \right| \right\} - \min \left\{ \left| \boldsymbol{\theta} (t) - \wh{\boldsymbol{\theta}} (t-1) \right| \right\} \right)^2 \right]}{\mathbb{E} \left[ \left\| \boldsymbol{\theta} (t) - \wh{\boldsymbol{\theta}} (t-1) \right\|_2^2 \right]},    
\end{align}
noting that 
\begin{align}\label{NoteThatUpDown}
\boldsymbol{\theta} (t) - \wh{\boldsymbol{\theta}} (t-1) = - \eta (t-1) \sum\limits_{m = 1}^M \sum\limits_{i=1}^{\tau} \frac{B_m}{B} \nabla F_m \left( \boldsymbol{\theta}_m^i (t-1), \xi_m^i (t-1) \right),    
\end{align}
and in (b) we define $\varepsilon \triangleq \max\nolimits_{t} \{\varepsilon(t)\}$.
According to (\ref{NoteThatUpDown}), from the convexity of $\left\| \cdot \right\|^2_2$, it follows that
\begin{align}
\mathbb{E} \left[ \left\| \boldsymbol{\theta} (t) - \wh{\boldsymbol{\theta}} (t-1) \right\|_2^2 \right] & \le \eta^2 (t-1) \sum\limits_{m = 1}^M \sum\limits_{i=1}^{\tau} \frac{B_m}{B} \mathbb{E} \left[ \left\| \nabla F_m \left( \boldsymbol{\theta}_m^i (t-1), \xi_m^i (t-1) \right) \right\|_2^2 \right] \nonumber\\
& \ainequalle  \eta^2 (t-1) \tau^2 G^2, 
\end{align}
where (a) follows from Assumption \ref{AssumpBoundedVarGradient}. Accordingly, (\ref{App_ThetaHatWithA}) reduces to
\begin{align}
\mathbb{E} \left[ \left\| \wh{\boldsymbol{\theta}} (t)  \right\|_2^2 \right] \le \mathbb{E} \left[ \left\| {\boldsymbol{\theta}} (t)  \right\|_2^2 \right] + \Big(\frac{\eta (t-1) \tau G}{2q_1(t)}\Big)^2 \varepsilon d.    
\end{align}
Substituting the above inequality into (\ref{App_ThetaHat_1}) yields
\begin{align}
\mathbb{E} \left[ \left\| \wh{\boldsymbol{\theta}} (t) - {\boldsymbol{\theta}}^* \right\|_2^2 \right] & \le  \mathbb{E} \left[ \left\| {\boldsymbol{\theta}} (t)  \right\|_2^2 \right] + \mathbb{E} \left[ \left\| {\boldsymbol{\theta}}^* \right\|_2^2 \right] - 2 \mathbb{E} \left[ \langle {\boldsymbol{\theta}} (t) , {\boldsymbol{\theta}}^* \rangle \right] + \Big(\frac{\eta (t-1) \tau G}{2q_1(t)}\Big)^2 \varepsilon d \nonumber\\
& = \mathbb{E} \left[ \left\| {\boldsymbol{\theta}} (t) - {\boldsymbol{\theta}}^* \right\|_2^2 \right] + \Big(\frac{\eta (t-1) \tau G}{2q_1(t)}\Big)^2 \varepsilon d.
\end{align}

\end{document}